\documentclass[letterpaper,pra,notitlepage,superscriptaddress, nofootinbib, showpacs, eqsecnum]{revtex4-1}

\usepackage {amsmath}
\usepackage{graphicx}
\usepackage{braket}
\usepackage{amsthm}
\usepackage{amssymb}
\usepackage{color}
\usepackage{mathtools}
\usepackage{tikz}
\usepackage{tensor}
\usetikzlibrary{decorations}
\usetikzlibrary{patterns}

\newtheorem{mydef}{Definition}

\newtheorem{fact}{Fact}

\usepackage[colorlinks=true,urlcolor=blue, linkcolor=blue]{hyperref}
\usepackage{enumitem}

\def\({\left(}
\def\){\right)}

\usepackage{bm}

\renewcommand{\vec}[1]{\bm{\mathrm{#1}}}

\graphicspath{ {./Figures/} }



\setlist[description]{leftmargin=0pt}

\renewcommand{\vec}[1]{\bm{\mathrm{#1}}}

\newcommand{\orcid}[1]{\href{https://orcid.org/#1}{\includegraphics[width=8pt]{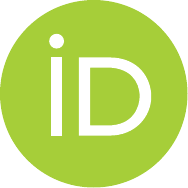}}}

\newcommand{\ketbra}[1]{\ket{#1}\bra{#1}}

\usepackage{xcolor}

\begin{document}
\title{Twirling and Hamiltonian Engineering via Dynamical Decoupling for GKP Quantum Computing}
\author{Jonathan Conrad \orcid{0000-0001-6120-9930}}\email{j.conrad1005@gmail.com}
\affiliation{Dahlem Center for Complex Quantum Systems, Physics Department, Freie
Universität Berlin, Arnimallee 14, 14195 Berlin, Germany}
\affiliation{Helmholtz-Zentrum Berlin für Materialien und Energie, Hahn-Meitner-Platz 1, 14109
Berlin, Germany}

\date{\today}
\begin{abstract}
I introduce an energy constrained approximate twirling operation that can be used to diagonalize effective logical channels in GKP quantum error correction, project states into the GKP code space and  construct a dynamical decoupling sequence with fast displacements pulses to distill the GKP stabilizer Hamiltonians from a suitable substrate-Hamiltonian. The latter is given by an LC-oscillator comprising a superinductance in parallel to a Josephson Junction. This platform in principle allows for protected GKP quantum computing without explicit stabilizer measurements or state-reset by dynamically generating a ``passively" stabilized GKP qubit.

\end{abstract}
\maketitle
%

\section{Introduction}
In our effort towards a scalable fault-tolerant quantum computer, bosonic error corrected quantum memories have recently gained much theoretical and experimental attention \cite{terhal_review, Grimsmo_2020, Campagne_Ibarcq_2020, Fluehmann_2019}. On the one hand this attention has been motivated by its resource efficiency as compared to standard qubit-based quantum error correction methods and its experimental viability using modern techniques, and on the other due to a promise of favourable properties when concatenated with more traditional qubit-based quantum error correcting codes \cite{Puri_2019, Puri_2020, vuillot2019quantum, Noh_2020_fault}. One particularly promising bosonic quantum error correcting code is the GKP code \cite{gottesman2001encoding} that encodes a qubit in a quantum mechanical oscillator via periodic, non-Gaussian eigenstates of the stabilizers 

\begin{equation}
S_p=D(\sqrt{2\pi})=e^{-i2\sqrt{\pi}\hat{p}} =\overline{X}^2,\hspace{1cm} S_q=D(\sqrt{2\pi})=e^{i2\sqrt{\pi}\hat{q}}=\overline{Z}^2,
\end{equation}
where the displacement operator is given as 
$
D(\alpha)=e^{\alpha \hat{a}^{\dagger}-\alpha^* \hat{a}}$ satisfying $D(\alpha)D(\beta)=e^{\omega(\alpha,\beta)}D(\beta)D(\alpha),
$
with complex symplectic form $\omega (\alpha,\,\beta)=\alpha \beta^*-\alpha^*\beta$, and the creation-and annihilation operators are defined via 
$
\hat{a}=\frac{\hat{q}+i\hat{p}}{\sqrt{2}},\;[\hat{a},\hat{a}^{\dagger}]=I.
$

The code states of the GKP code are by design especially well suited to correct against small displacement and have also been shown to be resilient against photon loss \cite{Albert_2018,terhal_review} which is a common source of errors in quantum harmonic oscillator systems.  Since the first scalable proposals for experimental implementation of the GKP code \cite{Terhal_2016}, it has also found place in applications beyond the stabilization of a single logical qubit --  such as metrology \cite{duivenvoorden2017single} or error corrected transmission of general bosonic states \cite{Noh_2020_encoding}.

Crucially, energy constraints limit how well code states can be obtained in practise, such that physical realizations can only be obtained \textit{approximately} in a form \cite{gottesman2001encoding, matsuura2019equivalence}  (up to normalization)
\begin{align} 
\ket{0_{\Delta}} &=\int_{\mathbb{R}} dq\;  \sum_{n \in \mathbb{Z}} e^{-2\Delta^2 \pi n^2 } e^{-\frac{1}{2 \Delta^2} (q-2n\sqrt{\pi})^2}\ket{q}, \label{eq:codestate_0}\\
\ket{1_{\Delta}}  &=\int_{\mathbb{R}} dq\;  \sum_{n \in \mathbb{Z}} e^{-2 \Delta^2 \pi n^2 }e^{-\frac{1}{2 \Delta^2} (q-(2n+1)\sqrt{\pi})^2}\ket{q},\label{eq:codestate_1}
\end{align}
with squeezing parameter $\Delta<1$ that characterises the quality of the approximation.
Given such encoded states, arbitrary fault-tolerant single-and two qubit gates can be implemented via standard Gaussian-operations and homodyne measurements on the quadratures of the oscillator \cite{Baragiola_2019, Yamasaki_2020}. Error correction can be performed via Steane-Type or Knill-Type error correction where modular quadrature displacements $q=\epsilon_q\mod\sqrt{\pi},\; p=\epsilon_p\mod\sqrt{\pi}$ are learned from stabilizer measurements and suitable correction displacements are applied \cite{gottesman2001encoding, Glancy_2006}. Alternatively, stabilizer measurement and correction can also be simulated by reservoir-engineering \cite{royer2020stabilization}, where an ancillary qubit is repeatedly entangled with the oscillator and reset. This approach is akin to a time-discretized version of known driven-dissipative engineering schemes for the \textit{cat-}code \cite{Mirrahimi_2014, Cohen_2017}.

Aside from active error correction implementations or the autonomous simulation thereof based on explicit dissipation mechanisms, implementations of the GKP stabilizer Hamiltonian 
\begin{equation}
H_{GKP} \propto -S_q-S_p +h.c. = -2 \cos(2\sqrt{\pi} \hat{p}) - 2 \cos(2\sqrt{\pi} \hat{q}) \label{eq:GKPHamiltonian}
\end{equation}
for passive error correction have also recently been proposed using a gyrator circuit  \cite{rymarz2020hardwareencoding} or phase-slip junction \cite{Le2020DoublyNS}. The form of the Hamiltonian implies a $2-$fold degenerate ground space which can be associated with the code space. For (qubit-) systems where the stabilizers have discrete spectrum, such stabilizer Hamiltonian yields a finite energy barrier to stabilize the qubit ground-space when the system is weakly coupled to a thermal bath at low temperature. For the perfect GKP-Hamiltonian, however, this is not immediately the case as the stabilizer have continuous spectrum; but an energy gap as consequence of an imperfect approximation to the GKP Hamiltonian may still arise \cite{rymarz2020hardwareencoding}.

The central theme of this paper will revolve around the notion of \textit{twirling}. Originally proposed in terms of \textit{state-twirling} as a tool for entanglement purification in \cite{Bennett_1996_purif,Bennett_1996_twirl}, it is most often discussed in its incarnation as a \textit{channel-twirl} which has been used widely in the simulation of error channels for quantum error correction and can be implemented in practice to remove coherent error build-up. Aiming at a practically feasible implementation of a Pauli-twirl for GKP encoded qubits, I design a suitable \textit{twirl-measure} which results in a logical Pauli-channel twirl that respects the bosonic code-degeneracy and can be tuned to adapt the quality of the twirl to energy constraints that the physical platform might have. The analysis of the twirl happens at the level of the \textit{chi-function}, which has proven to be a useful representation to analyse bosonic channels that GKP states traverse.
The same twirl-measure can be used to design an energy constrained GKP state-twirl to filter out unwanted state-contributions by acting as an approximate projector on the characteristic function of an input state.
\\
The commonality between a state- and a channel-twirl is its underlying structure of a \textit{unitary group projector} onto the commutant of a twirling group $\mathcal{G} \subset \mathcal{U}(d)$. In the case of a state-twirl, the density matrix is input to this projection; in the case of a channel twirl, it is the natural representation of the channel. A further effective implementation of a unitary group projector in physical systems exists in the framework of \textit{dynamical decoupling} (DD) \cite{Viola_1999, Zanardi_1999} where the Hamiltonian of a quantum system undergoes the group projection as an effect of fast, coherent control. I show how the defined twirling measure can be translated into a bang-bang periodic dynamical decoupling (BPD) sequence which acts as an approximate projector on the characteristic function of the Hamiltonian and propose a superconducting circuit involving a Josephson-Junction (JJ) which, after twirling, results in an approximation of the GKP Hamiltonian \eqref{eq:GKPHamiltonian}. The gapped spectrum and quality of GKP-Eigenstates of this effective Hamiltonian are studied numerically and I point out how the same techniques can be used to implement and tune a logical Pauli + stabilizer Hamiltonian when also a superconducting circuit element is present that restricts to $2-$cooper pair tunneling.

I note that this application of dynamical decoupling to Hamiltonian engineering is unorthodox as instead of trying to decouple unwanted error-interactions, dynamical decoupling in this proposal is used to distill targeted Hamiltonian terms from a \textit{substrate-Hamiltonian}. This technique can be considered a form of \textit{Floquet-} Hamiltonian engineering \cite{Oka_2019} which has already been studied to obtain Hamiltonians inaccessible in static systems.

The structure of the paper will be as follows. First I discuss the notions of characteristic functions of states and Hamiltonians and that of the chi-function that will be used throughout this paper in section \ref{sec:prelim} and point out examples and applications for every notion I introduce. This perspective will prove convenient when in section \ref{sec:twirl} the finite energy regularized GKP twirl is explained.
In section \ref{sec:Hamilton_engineering} I show how the twirl can be used to design a dynamical decoupling sequence using bang-bang displacements and apply the DD sequence to a suitable superconducting circuit to distill the GKP Hamiltonian. I discuss inequalities on the experimental design parameters necessary to realize the scheme proposed. I close by naming possible improvements to my scheme and related questions that require more theoretical investigation in the future.

\section{Preliminaries}\label{sec:prelim}

I make extensive use of the displacement operator basis for quantum states, -channels and Hamiltonians. This will prove as the most natural and useful choice since the GKP stabilizers are displacement operators and hence GKP QEC naturally corrects against small displacements. I adopt the convention $\hbar=1$ wherever absolute units are of no relevance.

The displacement operators form an operator basis on the bosonic Hilbert space, satisfying  orthogonality \cite{gottesman2001encoding}
\begin{equation}
\mathrm{Tr} D^{\dagger}(\alpha)D(\beta)=\pi \delta^2 (\alpha -\beta ), \label{eq:displacement_orthogonal}
\end{equation}
such that every operator $F$ can be expressed in terms of its characteristic function $f(\alpha)=\mathrm{Tr}D^{\dagger}(\alpha) F$ as 
\begin{equation}
F=\frac{1}{\pi}\int d^2 \alpha\; f(\alpha) D(\alpha). \label{eq:O_expansion}
\end{equation}

\paragraph{States} The characteristic function of a state $\rho(\alpha):=TrD^{\dagger}(\alpha) \rho =\rho^*(-\alpha),\; \rho(0)=1$ is what is conventionally denoted as \textit{the} characteristic function. To ease communication, I will regard all complex valued coefficients of operators in the displacement basis as ``characteristic function".
A well known quantity related to the characteristic function of a state is the Wigner function,
\begin{equation}
W(\alpha)=\pi^{-2}\int d^2\beta \, e^{\omega(\alpha,\beta)} \rho(-\beta),\label{eq:Wigner}
\end{equation}
which serves as quasiprobability distribution to visualize states and interpret their support in phase space via the association $\alpha=\frac{q+ip}{\sqrt{2}}$. 

The characteristic function yields simple expressions for GKP stabilizer- and logical Pauli expectation values and can be used to quantify the quality of GKP states by adopting the definition of effective squeezing parameters in \cite{Weigand_2020}
\begin{equation}
\Delta_q=\sqrt{\frac{-1}{\pi}  \mathrm{ln}\left( {|\rho\left(i\sqrt{2\pi}\right)|}\right)} ,\; \Delta_p=\sqrt{\frac{-1}{\pi}\,\mathrm{ln}\left( {|\rho\left(\sqrt{2\pi}\right)|}\right)},
\end{equation}
which are invariant under displacement. These effective squeezing parameters can be associated with the parameter $\Delta$ in \eqref{eq:codestate_0},\eqref{eq:codestate_1}, which grow rapidly with decreasing $|\rho\left(\sqrt{2\pi}\right)|,\, |\rho\left(i\sqrt{2\pi}\right)|$ and reflect the amplitude of the $\sqrt{2\pi}$ periodic Fourier component of the Wigner function.

\paragraph{Hamiltonians}
Similar to the representation of states, one may also represent a Hamiltonian via its characteristic function $h(\alpha)=\mathrm{Tr}D^{\dagger}(\alpha) H$, which satisfies $h^*(-\alpha)=h(\alpha)$. 
Examples of Hamiltonians admitting a particularly simple characteristic function are the passive GKP Hamiltonian \eqref{eq:GKPHamiltonian} with
\begin{equation}
h_{GKP}(\alpha)= - \delta^2(\alpha-\sqrt{2\pi})-\delta^2(\alpha+\sqrt{2\pi})-\delta^2(\alpha-i\sqrt{2\pi})-\delta^2(\alpha+i\sqrt{2\pi}), \label{eq:h_GKP}
\end{equation}
the \textit{2-legged cat} qubit Hamiltonian $H\propto -\sigma_L^z$, where for ($|\beta|\gg 1$)
\begin{align}
\sigma_L^z&=\ketbra{C_\beta^+}-\ketbra{C_\beta^-}=2\ket{\beta}\bra{-\beta}+2\ket{-\beta}\bra{\beta}, \;
\ket{C_\beta^{\pm}}=N_{\beta}^{-1}(\ket{\beta}\pm \ket{-\beta}),
\end{align} 
it holds that
\begin{equation}
h_{cat,\, 2}(\alpha) \propto -e^{-\frac{1}{2}|\alpha-2\beta|^2}-e^{-\frac{1}{2}|\alpha+2\beta|^2}.
\end{equation}
Similarly, for passive encoding of the  \textit{4-legged cat} \cite{Cohen_2017} one aims at implementing a characteristic function
\begin{equation}
h_{cat,\, 4}(\alpha) \propto -e^{-\frac{1}{2}|\alpha-2\beta|^2}-e^{-\frac{1}{2}|\alpha+2\beta|^2}-e^{-\frac{1}{2}|\alpha-i2\beta|^2}-e^{-\frac{1}{2}|\alpha+i2\beta|^2}.
\label{eq:4cat}
\end{equation}
The 4-legged cat and GKP Hamiltonians are visualized in fig. \ref{fig:hchar} $(a)$.

\begin{figure}
\center
\includegraphics[width=.5\textwidth]{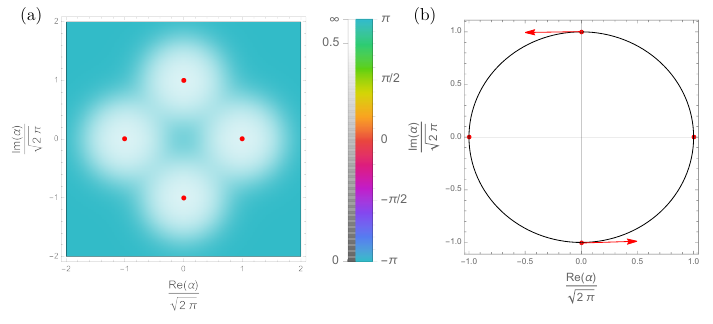}
\caption{$(a)$ The characteristic function $h_{cat,\, 4}(\alpha)$ with $\beta=\sqrt{\frac{\pi}{2}}$ and in red the position of the delta peaks of the GKP stabilizer Hamiltonian. $(b)$ the characteristic function $h_{JJ}(\alpha;t)$  with $\varphi=\sqrt{2\pi}$ traverses the indicated circle with time. For $\omega t\geq 2\pi$ the rotating points are smeared out over the circle which represents the first order RWA $\overline{h}^{(1)}_{JJ}(\alpha; \,t)$. The delta peaks of the characteristic function $h_{GKP}(\alpha)$ are indicated in red.} \label{fig:hchar}

\end{figure}

\begin{figure}
\center
\centering
\includegraphics[width=.3\textwidth]{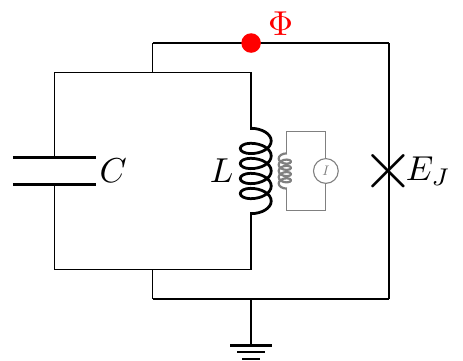}
\caption{Quantum harmonic oscillator comprising a cavity and a (super-)inductance coupled to a Josephson Junction. In gray a circuit element is indicated to implement displacements on the oscillator via inductive coupling.}\label{fig:circuit}
\end{figure}

An important Hamiltonian for this paper will be that of a quantum harmonic oscillator with frequency $\omega=\sqrt{LC}^{-1}$ coupled to a Josephson junction (JJ) as illustrated in fig. \ref{fig:circuit}. The characteristic function of the JJ Hamiltonian is given by

\begin{equation}
h_{JJ}(\alpha)=Tr[D^{\dagger}(\alpha)H_J]=\frac{-E_J}{2}\left\{\delta^{2}(\alpha-i\varphi)+\delta^{2}(\alpha+i\varphi) \right\}, \label{eq:hJJ}
\end{equation}
where $E_J$ is the Josephson energy and $\varphi=\sqrt{\frac{\pi Z}{R_Q}}$ where $Z=\sqrt{\frac{L}{C}}$ is the impedance of the cavity mode as seen by the junction and $R_Q=\frac{h}{(2e)^2}$ the resistance quantum.

In a frame co-rotating with the bare oscillator the effective Hamiltonian characteristic function for $H_{rot}(t)=U_0^{\dagger}HU_0+i\frac{d U_0^{\dagger}}{dt}U_0$ reads 
\begin{equation}
h_{rot}(\alpha;\, t)=h_{int}(\alpha e^{-i\omega t}).
\end{equation}
The unitary evolution of the system in this frame is given by the Magnus expansion
\begin{equation}
U(t)=\mathcal{T}\exp{\left(-i\int_0^{t} H_{rot}(t') dt\right)}= \exp{\left(-i \overline{H}(t)\right)}, \label{eq:Magnus}
\end{equation}
for which term-wise characteristic functions  of $\overline{H}(t)=\sum_k \overline{H}^{(k)}(t)$ can be evaluated to

\begin{align}
\overline{h}^{(1)}(\alpha; \,t)&=\int_0^t h_{int}(\alpha e^{-i\omega t'}) dt', \label{eq:h_first}\\
\overline{h}^{(2)}(\alpha; \,t)&=-i\int d^2\beta \int_0^t dt' \int_0^{t'}dt'' h_{int}((\alpha-\beta ) e^{-i\omega t'}) h_{int}(\beta e^{-i\omega t''}) \sin{\left( Im \left(\alpha \beta^* \right) \right)}, \label{eq:h_second} \\
\overline{h}^{(3)}(\alpha; \,t)&=...
\end{align}

The first order term, written out as
\begin{align}
\overline{H}^{(1)}(t)
&=\int d^2\alpha \left[\int_0^t h_{int}(\alpha e^{-i\omega t'}) dt'\right] D(\alpha) \\
&=\int d\phi_{\alpha} d|\alpha| |\alpha| h_{int}(|\alpha| e^{i\phi_{\alpha}} ) \left[\int_0^t  D(|\alpha| e^{i(\phi_{\alpha}+\omega t' )})dt'\right],
\end{align}
where the integration parameter is written as $\alpha=|\alpha|e^{i\phi_{\alpha}}$, corresponds to a phase-average of the characteristic function.
For large $\omega t\geq 2\pi$ the phase average of the characteristic function can be understood as a rotated \textit{smearing} in phase space, such that the phase information of the characteristic function is integrated out $\int_0^t h_{int}(\alpha e^{-i\omega t'}) dt' =h_{int}(|\alpha|)$ and  produces an operator diagonal in the Fock basis. 

This can also be seen directly by evaluating the displacement operator in the Fock basis \cite{Cahill_Glauber, Cohen_2017}
\begin{align}
\int_0^t dt' \braket{m|D(|\alpha| e^{i(\phi_{\alpha}+\omega t' )})|n} 
&\xrightarrow{\omega t \geq 2\pi}  \delta_{m,n} L_n(|\alpha|^2) e^{-\frac{|\alpha|^2}{2}} t,
\end{align}
where $L_n(\cdot)$ are the  Laguerre polynomials 

In first order RWA  the Hamiltonian in \eqref{eq:hJJ} is given by
\begin{equation}
\overline{h}^{(1)}_{JJ}(\alpha; \,t)= -E_J\delta(|\alpha|-\varphi) t \label{eq:hJJ_RWA}
\end{equation}
which is a good approximation when $\hbar \omega \gg E_J$ and $\omega t \geq 2\pi$, such that higher order terms to the Magnus expansion are negligible \cite{Cohen_2017} and the phase information is fully smeared out.

\paragraph{Channels}
Every CP superoperator $\mathcal{N}$ on an oscillator Hilbert space admits a contiuous \textit{chi}-function representation 
\begin{align}
\mathcal{N}(\rho) = \iint d^2\alpha d^2\beta c(\alpha, \beta)D(\alpha)\rho D(\beta)^{\dagger}, \label{eq:cont_chi}\\
\int d^2\beta c(\alpha +\beta, \beta) e^{\omega(\alpha, \beta)/2} \leq\delta^{2}(\alpha).   \label{eq:cont_chi_norm}
\end{align}
This is because CP maps $\mathcal{N}$ admit the Kraus form
\begin{equation}\mathcal{N}(\rho)=\sum_l E_l \rho E_l^{\dagger}\hspace{1cm}  \text{with}\hspace{1cm} \sum_l E_l^{\dagger}E_l \leq 1 \label{eq:Kraus}.
\end{equation} 
Equality in the last equation is given when the map is also TP.  Applying \eqref{eq:O_expansion} to each Kraus operator yields \eqref{eq:cont_chi} with 
\begin{equation}
c(\alpha,\beta) =\pi^{-2}\sum_l c_l(\alpha)c_l^*(\beta), \hspace{1cm} c_l(\alpha)=Tr[D^{\dagger}(\alpha) E_l]. \label{eq:cont_chi_coeff}
\end{equation}
Similarly, \eqref{eq:cont_chi_norm} is derived. \eqref{eq:cont_chi_coeff} also implies that the diagonal elements of  $c(\alpha,\beta)=c^*(\beta, \alpha)$ are real-valued.

For coherent channels, $c(\alpha, \beta)$ factorizes. This is for example the case for the finite squeezing approximation to GKP states, see \cite{gottesman2001encoding}, where 
\begin{equation}
c(\alpha, \beta)=\frac{1}{\pi \Delta^2} e^{-\frac{|\alpha|^2}{\Delta^2}} e^{-\frac{|\beta|^2}{\Delta^2}}. \label{eq:finite_squeezing_c}
\end{equation}
To draw an example for incoherent noise, I evaluate the chi-function for the photon loss channel in appendix \ref{appendix:photonloss}, for which the Kraus operators are given as 
\begin{equation}
\hat{E}_l =\left(\frac{\gamma}{1-\gamma}\right)^{\frac{l}{2}}\frac{\hat{a}^l}{\sqrt{l!}}(1-\gamma)^{\hat{n}}, \gamma =1-e^{-\kappa t},
\end{equation}
with loss rate $\kappa$. The chi-function takes the form
\begin{equation}
c^{\gamma}(\alpha,\beta)=\left(\frac{\overline{\gamma}}{\pi}\right)^2\braket{\beta|\alpha}^{2\overline{\gamma}-1}, \label{eq:c_photonloss}
\end{equation}
where $\braket{\beta|\alpha} $ denotes the inner product between coherent states at phase-space positions $\alpha,\,\beta$ and I have introduced the effective loss parameter
\begin{equation}
\overline{\gamma}=\frac{1}{1-\sqrt{1-\gamma}}=\frac{2}{\gamma}-\frac{1}{2}-\frac{\gamma}{8} +O(\gamma^2)\in [1,\infty). 
\end{equation}
From \eqref{eq:c_photonloss} it can be seen that coherent, i.e. off-diagonal, contributions to the chi-function are exponentially suppressed in the distance $|\alpha-\beta|$, an effect that is amplified for smaller $\gamma \approx \kappa t$. Diagonal elements of the chi function have constant amplitude, which emphasizes the non-classicality of this channel.

\section{Regularized Twirling}\label{sec:twirl}

\subsection{Channel and -state twirling}
The notion of twirling was first introduced in \cite{Bennett_1996_purif,Bennett_1996_twirl} as tool for entanglement purification for a two-qubit state. 
The basic idea is to draw elements of a unitary subgroup $G$ at random and apply them to an input state $\rho$,

\begin{equation}
\Pi(\rho)=\frac{1}{|G|}\sum_{g\in G} g \rho g^{\dagger}.
\end{equation}
As a consequence of the group properties, it can be shown that $\Pi$ acts as projector onto the commutant $\mathcal{C}G$ of $G$, which is for instance non-trivially spanned by the qubit-Bell states if $G=\{g\otimes g | g \in \langle X, Z \rangle  \}$ where  $X,\,Z $ are the Pauli matrices \cite{Bennett_1996_twirl}. 
This idea is extended to the twirling of qubit-channels by considering the projector over groups of the form $T=\{ g\otimes \overline{ g} | g \in G \}$ acting on the natural representation of a channel (with Pauli-basis chi-matrix $\chi_{\alpha\beta}$) \begin{equation}
\hat{N}=\sum_{\alpha,\beta \in \{I,x,y,z\}} \chi_{\alpha\beta} \sigma_{\alpha}\otimes \overline{\sigma}_\beta.
\end{equation}
Choosing $G$ to be the single qubit Pauli group and using that Paulis either commute or anti-commute, it can be seen that such twirl effectively renders $\chi_{\alpha\beta}$ diagonal. Physically the twirled channel is obtained by the following sequence:\; 
$1.$ draw a unitary $g\in G$ at random and apply it to the input state,
$2.$ apply the channel,
$3.$ apply the $g^{\dagger}$ to the output of the channel.
When averaged over the channel outputs, the average channel will correspond to the twirled one. Pauli Twirling for qubit systems has been studied extensively in the literature, see e.g. \cite{Cai_2019, katabarwa2017dynamical} and references therein. Twirling for channels acting on GKP encoded qubits has also been proposed \cite{wang:ms-thesis}, but so far necessitates the implementation of infinitely large displacements with large probability. While this proves useful to approximate channels for numerical simulations, it is not a physical operation that can be implemented. In the following, to design an approximate logical twirl for a GKP-encoded qubit I will deviate from the standard design of twirls in two ways: Firstly, the twirling group will not just range over the logical Pauli gates, but also include their logicals equivalents to actively take advantage of the fact that logical GKP encoded gates differing by stabilizers correspond to different physical operations. Furthermore, I show how the probability distribution over the twirling group can be adapted to tune the strength of the twirl, which is necessary since what I construct in the following is an \textit{approximate} twirl. Equipped with these ingredients, I show how they can be used for state twirling by logical- and stabilizer operations. 

\begin{mydef}
A \textit{displacement-}twirling operation $\tau$ with respect to (twirl-) measure $\mu$ is a CPTP-preserving map on superoperators, that maps a superoperator as in \ref{eq:cont_chi} to
\begin{align}
\tau \circ \mathcal{N}(\rho)
&=\int d\mu(\gamma) D(\gamma)^{\dagger}\mathcal{N}\left(D(\gamma)\rho D^{\dagger}(\gamma)\right)D(\gamma)\\
&=\int d\mu(\gamma)\iint d^2\alpha d^2\beta c(\alpha, \beta) D(\gamma)^{\dagger}D(\alpha) D(\gamma)\rho D^{\dagger}(\gamma)D(\beta)^{\dagger}D(\gamma) \nonumber \\
&=\iint d^2\alpha d^2\beta c(\alpha, \beta) \left[\int d\mu(\gamma) e^{\omega (\alpha -\beta, \gamma)} \right] D(\alpha)\rho D(\beta)^{\dagger}. \label{eq:twirl}
\end{align}
\end{mydef}
Repeated twirling effectively acts as a map on the chi-function 
\begin{equation}
\tau^N:\; \;c(\alpha, \beta) \mapsto  c_{\tau,N}(\alpha, \beta)=c(\alpha, \beta) \left[\int d\mu(\gamma) e^{\omega (\alpha -\beta, \gamma)} \right]^N. \label{eq:twirl_n}
\end{equation}
For the uniform measure, as considered for twirling the GKP finite-squeezing error in \cite{wang:ms-thesis, vuillot2019quantum}, $d\mu(\gamma)=d^2\gamma$, the chi-function twirls to
\begin{equation}
c_{\tau}(\alpha, \beta)=c_{\tau,1}(\alpha, \beta)=\pi^2 c(\alpha, \beta) \delta^{2}(\alpha-\beta).
\end{equation}

Uniform displacement-twirling hence renders any channel a stochastic displacement channel. As I pointed out, this is not a physical operation since arbitrarily large displacements are involved with high probability. A state undergoing a uniformly twirled channel hence would need to be pumped up to infinite average photon number in intermediate steps of the channel, rendering this approach unfeasible in practice. 
The same argument applies to the state-twirling operation using stabilizer-shifts considered in \cite{Noh_2020_fault}.

In the following I give an implementation of a regularized, feasible channel-twirl via displacements that respects the GKP code-degeneracy.
The aim is modest, instead of asking for an effectively diagonal chi-function, I aim to obtain a logical Pauli-channel (similar to twirling channels for a regular qubit). To this end consider the measure 
\begin{equation}
d\mu(\gamma)=\frac{1}{16}\sum_{m,n \in \{-1,0,0,1\}} \delta^{2}\left(n\sqrt{\frac{\pi}{2}}+im\sqrt{\frac{\pi}{2}}-\gamma\right) d^2\gamma. \label{eq:measure}
\end{equation}
\begin{figure}
    \centering
    \includegraphics{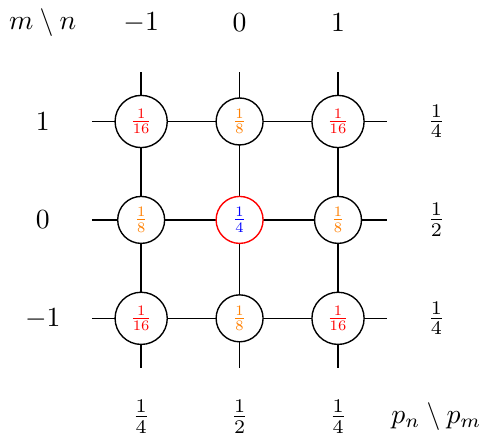}
    \caption{Starting at the red node, the transition probabilities by the local moves in the random walk $X$ are indicated. The transition probabilities are given by the joint probabilities of the $1D$ random walks along each axis.}
    \label{fig:rdwalkX}
\end{figure}
Note that the sums each explicitly involve two $0$'s. The resulting local probability distribution of random displacements is summarized in fig. \ref{fig:rdwalkX}.
With this choice of measure, \eqref{eq:twirl_n} becomes 
\begin{equation}
c_{\tau,N}(\alpha, \beta)=c(\alpha, \beta) \left[\frac{1}{4} \left(1+\cos\sqrt{2\pi}\mathrm{Re}(\alpha-\beta)\right) \left(1+\cos\sqrt{2\pi}\mathrm{Im}(\alpha-\beta)\right) \right]^N. \label{eq:log_twirl}
\end{equation}
In the limit of large $N$, the twirl enforces that $\mathrm{Re}(\alpha-\beta)\, mod\, \sqrt{2\pi}=\mathrm{Im}(\alpha-\beta)\, mod\, \sqrt{2\pi}=0$ for all $\alpha,\,\beta$ that non-trivially support $c_{\tau, N}$ in \eqref{eq:log_twirl}, that is $\alpha$ and $\beta$ will only differ by stabilizer shift.. Furthermore, if the support of $c(\alpha, \beta)$ is sufficiently narrow (i.e. compact on a radius much smaller than $\sqrt{2\pi}$ on either of the arguments), it can further be argued that only the stochastic terms $c(\alpha, \alpha)$ prevail (similar to the argument in \cite{Noh_2020_fault}). The factor $c_{\tau,N}(\alpha, \beta)/c(\alpha, \beta)$ for $N=1,10$ is plotted in fig. \ref{fig:factor_plot}.
\begin{figure}
\center
\centering
\includegraphics[width=.5\textwidth]{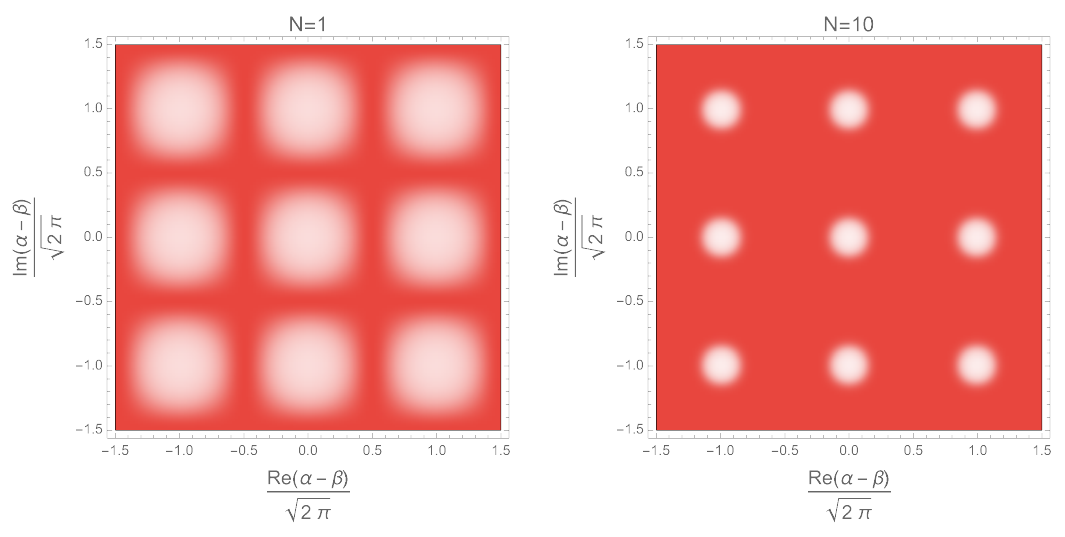}
\caption{The factor $c_{\tau,N}(\alpha, \beta)/c(\alpha, \beta)$ for $N=1,10$. The peaks sharpen with increasing $N$ and are separated by $\sqrt{2\pi}$. Colorbar is as in fig. \ref{fig:hchar}.}\label{fig:factor_plot}
\end{figure}

\paragraph{Recompiling $\tau^N$}
For $N$ repetitions of the local twirl,  the displacements of consecutive rounds can be combined and  indexed by $(n,m) \in \{-N,..,N\}^2$. The resulting probability distribution of displacements can be interpreted as one of a $N$-step random walk starting at $(0,0)$ on a square lattice, where the vertices represent the labels $(n,m)$  (see also fig. \ref{fig:rdwalkX}). In appendix \ref{appendix:recompile} I derive the  probability to end up on vertex $(n,m)$ after $N$ steps, given by
\begin{equation}
P_X^N(n,m)=2^{-4 N}\binom{2N}{n+N}\binom{2N}{m+N},
\end{equation}
which gives an expression for the $N-$level twirl-measure
\begin{equation}
d\mu_N(\gamma)=\sum_{n, m =-N}^N P_X^N(n,m) \delta^{2}\left(n\sqrt{\frac{\pi}{2}}+im\sqrt{\frac{\pi}{2}}-\gamma\right) d^2\gamma. \label{eq:measure_N}
\end{equation}
Thus, by sampling the displacement indices $(n,m)$ from $P_X^N(n,m)$ in \eqref{eq:dispN_prob}, one can obtain an effective $N$-step twirl $\tau^N$ in a single shot. Note that (in the limit $N \rightarrow \infty$) what is considered here is a \textit{logical} twirl, such that the resulting channel behaves as a logical Pauli-channel. This is however done in a manner that respects the logical degeneracy (i.e. equivalence up to stabilizer shifts). \\
The maximum gain in average photon number in the intermediate step of the twirled channel is $\Delta \langle n \rangle_{max} =\pi N^2$,
which will be obtained with probability P(extremal)=$2^{2-4N}$.

\subsection{State twirling}

\paragraph{The completely mixed logical state}
The $N-$step twirling measure derived in this section can also be employed to twirl a state as opposed to a channel. To this end consider the action on the characteristic function of the state. Depolarizing a state with initial characteristic function $\rho$ with the previously used measure $\mu$ gives 
\begin{align}
\rho \mapsto \Pi(\rho) &=\int d\mu(\gamma) D(\gamma) \rho D^{\dagger}(\gamma)   \label{eq:state_twirl}\\
&=\int d^2\alpha \rho (\alpha)\left[\int d\mu(\gamma) e^{\omega (\alpha, \gamma)} \right] D(\alpha).
\end{align}
Applying the $N$-level twirl with respect to the recompiled measure yields 
\begin{equation}
\Pi^N:\;\; \rho(\alpha)\mapsto \rho (\alpha)\left[\frac{1}{4} \left(1+\cos\sqrt{2\pi}\mathrm{Re}(\alpha)\right) \left(1+\cos\sqrt{2\pi}\mathrm{Im}(\alpha)\right) \right]^N.    
\end{equation}
That is, in the limit of large $N$, the displacement coefficients will be supported compactly on $\{\alpha \in \mathbb{C}: \mathrm{Re}(\alpha) \mathrm{mod} \sqrt{2\pi}=
 \mathrm{Im}(\alpha) \mathrm{mod} \sqrt{2\pi}=0 \}$. A $\sqrt{2\pi}$ $2D-$ translationally invariant characteristic function implies that the Wigner function is supported on a lattice with periodicity $\sqrt{\frac{\pi}{2}}$, which in turn translates to $\sqrt{\pi}$-shift periodicity of $q,\, p$. The resulting state is a logical GKP completely mixed-state when the input state had any operator support on the GKP-lattice. Since its periodic structure is present, such completely logically depolarized state may in principle be used as ancillary state to realize modular GKP stabilizer measurements in scenarios where logical entropy feedback from the ancillary oscillator is irrelevant, e.g. in a single-shot metrological setup as in \cite{Noh_2020_encoding}, and as such may serve as a way to ``recycle" some states at the end of a computation.

\paragraph{Stabilizer Twirling}
In \cite{Noh_2020_fault} a state-twirling by powers of stabilizer-shifts was used to motivate the stochastic displacement approximation to finite squeezing errors. Since their approach involved a uniform measure over all possible stabilizer shifts, it can only be regarded as analytical tool.
To obtain a physically feasible implementation, I adapt my scheme to their setting by doubling the displacement shift-lengths. The probability distribution over the random walk on the, by a factor of $2$ stretched, lattice remains the same $P_X^N(n,m)$ as above.
In this case, in the limit of large $N$, the displacement coefficients will be supported compactly on $\{\alpha \in \mathbb{C}: \mathrm{Re}(\alpha)\, \mod\,  \sqrt{\frac{\pi}{2}}=
 \mathrm{Im}(\alpha) \mathrm{mod} \sqrt{\frac{\pi}{2}}=0 \}$, which translates to a $2\sqrt{\pi}$-shift periodicity of $(q,p)$ in real (phase-)space, 
\begin{equation}
\Pi_{stab.}^N:\;\; \rho(\alpha)\mapsto \rho (\alpha)\left[\frac{1}{4} \left(1+\cos 2\sqrt{2\pi}\mathrm{Re}(\alpha)\right) \left(1+\cos 2\sqrt{2\pi}\mathrm{Im}(\alpha)\right) \right]^N.  \label{eq:twirl_stab} 
\end{equation}
In contrast to the previous logical twirl, this twirl preserves the GKP-logical information. However, since $\rho(\sqrt{2\pi}),\,\rho(i\sqrt{2\pi})$ and $\langle Z\rangle=\rho\left(i\sqrt{\frac{\pi}{2}}\right) $ are preserved under this operation (assuming no errors during twirling), the effective squeezing parameter and logical Pauli expectation values remain invariant.
The average photon-number gain after the $N$-level twirl (modified by a factor $\frac{1}{4}$ for the logical state-twirl) is given by $\Delta \langle n \rangle_{state}=2\pi N$, such that twirling-level $N$ directly characterizes the average energy gain of the system.

\section{Dynamical engineering of GKP qubits} \label{sec:Hamilton_engineering}
\subsection{From Twirling to Dynamical decoupling} \label{sec:twirl_dd}
%

In addition to applying the twirl to channels and states, it can be shown that the framework of \textit{bang-bang  periodic Dynamical decoupling} (BPD) realizes a twirl of Hamiltonians. In the following I will briefly recap the essentials off  BPD, guided by the presentations in \cite{Viola_1999, Zanardi_1999, lidar_brun_2013, average_hamiltonian_theory}, to then point out how the twirl designed earlier can effectively be used to produce a system which stroboscopically evolves via the twirled Hamiltonian. This will be realized by interleaving the free system evolution with displacements applied instantaneously at time-steps dictated by the probability of the corresponding displacement from the $N-$level twirl-measure \eqref{eq:measure_N}.

Let the evolution of a quantum harmonic oscillator be guided by a Hamiltonian 
\begin{equation}
H(t) = H_0 + H_C(t),
\end{equation}
where  $H_C(t)$ describe control pulses applied to the system which is otherwise described by Hamiltonian $H_0$. 
In a frame co-rotating with the control evolution, also called the \textit{toggling frame}, the effective time evolution is then 
\begin{equation}
\tilde{U}(t)=U_C^{\dagger}(t)U(t), \hspace{.5cm} \frac{d}{dt}\tilde{U}(t)=-i \left[U_C^{\dagger}(t) H_0 U_C(t) \right]\tilde{U}(t)=-i\tilde{H}(t)\tilde{U}(t).
\end{equation}
Assuming periodic control $U_C(t+FT_C)=U_C(t),\; \forall F \in \mathbb{N}_0,$ the $T_C$ periodicity of the control is inherited by the effective Hamiltonian $\tilde{H}(t)$ that guides the evolution in the toggling frame. As result (assuming $U_C(0)=I$), the \textit{stroboscopic} effective evolution, as seen by only probing the system at times that are a multiple of the period length $FT_C$ is given by
$
     \tilde{U}(FT_C) =\tilde{U}(T_C)^F.
$
Using the Magnus expansion \eqref{eq:Magnus}, \cite{magnus, average_hamiltonian_theory}  the effective time evolution over the time unit $T_C$ can be represented as the time evolution via an effective, \textit{time-independent} Hamiltonian $\overline{H}_{av}=\frac{1}{T_C}\overline{H}(T_C)$
where, as in \eqref{eq:Magnus} $\overline{H}_{av}$ is given by an infinite series 
\begin{equation}
    \overline{H}_{av}=\sum_j \overline{H}^{(j)}_{av} \label{eq:Floquet_Magnus}
\end{equation}
with terms labeled by $(j)$ consist of $j$ nested commutators \cite{Iserles99expansionsthat} and are of order $O(T_C^j)$
\begin{align}
    \overline{H}^{(0)} _{av}&= \frac{1}{T_C}\int_0^{T_C} \tilde{H}(t') dt' , \label{eq:dd1}\\
    \overline{H}^{(1)}_{av} &= \frac{-i}{2T_C}\int_0^{T_C}\int_0^{t''}  [\tilde{H}(t''), \tilde{H}(t')] dt' dt'', \\
    \overline{H}^{(2)}_{av} &= ...
\end{align}
In BPD control is implemented by discrete instantaneous control pulse sequence $\{ P_k,\Delta t_k\}$, consisting of $M$ pulses $P_k$ following the time-intervals $\Delta t_k= \tau_k T_C$ of system evolution with $H_0$. The pulses are assumed to satisfy $P_0=I$ and $\prod_k P_k=I$. 
Specializing eq. \eqref{eq:dd1} to this setting it can be checked that the first order average time evolution is given by 
\begin{equation}
     \overline{H}^{(0)}=\sum_{k=1}^{M} \tau_k Q_k^{\dagger} H_0 Q_k, \label{eq:dd2}
\end{equation}
with the accumulated control pulse $Q_k=P_{k-1}P_{k-2}...P_1$, with $P_k=Q_{k+1}Q_k^{\dagger}$.

When the control pulses are chosen such that $Q_k$ as in \eqref{eq:dd2} correspond to multiples of logical displacements  for which $\tau_k$ yields the $N-$level twirl probability, i.e. by choosing  
\begin{align}
    \tau_{k}&=\tau_{k(n,m)}=P_X^N(n,m) \\
    Q_k &=Q_{k(n,m)}=D\left((n+im)\sqrt{\frac{\pi}{2}}\right),
\end{align}
the non-equidistant BPD sequence has a twirling of the Hamiltonian as effect, where the 2D displacements are labelled by $\{1,..,M=(2N+1)^2\}$ and $k(\cdot,\,\cdot)$ determines the order of the displacements.

\paragraph*{Control path ordering} To minimize the experimental effort of implementing the control pulses $P_k=Q_{k+1}Q_k^{\dagger}$ and to maintain $\prod_k P_k=I$, I construct a \textit{control graph}:  the vertices  $(n,m)$ of the control graph label the (accumulated) displacement amplitudes $\{Q_k\}$ (similar to fig. \ref{fig:rdwalkX}), and edges represent the allowed transitions, that is choices of $\{P_k=Q_{k+1}Q_k^{\dagger}\}$ to map between different accumulated control pulses. To minimize the necessary displacement amplitude at each instance I choose the edge connectivity as in a  \textit{kings graph} $C=(V=\{(n,m)\},E_{king})$, which is known to have Hamiltonian cycles for each $N$. The ordering $k(\cdot,\cdot)$ is then given by a  Hamiltonian cycle on the vertices of $C$ starting at $k(0,0)=1$. This construction ensures that each instantaneous control pulse displacement amplitude is bounded by a constant $|\xi|\leq \sqrt{\pi}$. For an example of a possible control sequence for $N=1$ see fig. \ref{fig:Hamiltonian_cycle}.

\begin{figure}
\includegraphics[width=.25\textwidth]{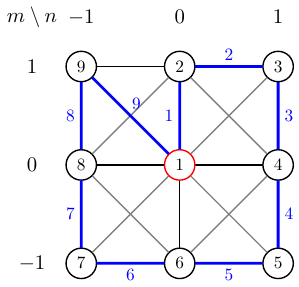}
\caption{One possible ordering of the control path as given by a Hamiltonian cycle on the control graph. Each vertex is associated with the accumulated control pulse $Q_k$ at time-index $k=k(n,m)$ and each edge with the instantaneous control pulses $P_k$ applied at the respective time step $k$ as indicated by the time labels on the edges.} \label{fig:Hamiltonian_cycle}
\end{figure}

The previously considered $N-step$ state twirl is thus mapped to an open-loop control sequence consisting of $M=(2N+1)^2$ displacements, where the relative time-interval lengths are given by the discrete state twirling measure. Stroboscopically, in first order Magnus expansion, this approximately projects the system Hamiltonian onto stabilizer displacements. I.e. given a system Hamiltonian $H_0$ with characteristic function $h_0(\alpha)$, the effective characteristic function under $N-$level BPD will be given by
\begin{equation}
\overline{h}_{av}^{(0)}= h_0 (\alpha)\left[\frac{1}{4} \left(1+\cos \sqrt{2\pi}\mathrm{Re}(\alpha)\right) \left(1+\cos \sqrt{2\pi}\mathrm{Im}(\alpha)\right) \right]^N.  \label{eq:twirl_hamiltonian}
\end{equation}
The DD sequence approximately \textit{filters out} Hamiltonian terms that are not close to stabilizer shifts, which are exactly retrieved in the Hamiltonian for $N \rightarrow \infty$. 

If the system Hamiltonian is chosen such that that its characteristic function has support on the GKP grid, this provides a way to dynamically engineer the GKP Hamiltonian \eqref{eq:h_GKP}. Similarly, a sequence employing stabilizer shifts can be used to distill GKP logical Hamiltonians to realize continuous logical Pauli rotations or to initialize a system via thermalization into logical Pauli-Eigenstates.

\subsection{Realization of dynamically  protected GKP states in superconducting circuits}
With the above machinery, it is possible to obtain effective average Hamiltonians that resemble the passive GKP Hamiltonian \eqref{eq:GKPHamiltonian} using the circuit shown in fig. \ref{fig:circuit}, of which the Hamiltonian in the rotating frame for $\hbar\omega\gg E_J$ is well supported on the four points in phase space $\alpha\in \{\pm\sqrt{2\pi},\, \pm i \sqrt{2\pi}\}$. For this purpose, the impedance of the cavity needs to be tuned to $Z=2R_Q$ such that $\varphi=\sqrt{2\pi}$. \footnote{Alternatively, any Hamiltonian whose characteristic function is supported on all four stabilizer displacements (red dots in figure \ref{fig:hchar}) would do. Trivially, the Hamiltonian $H_{\pi}=-(-1)^{\hat{n}}$, for which $h(\alpha)=const.$  \cite{Royer_1977} and the  \textit{4-legged cat} Hamiltonian  $H_{4, cat}$ with $\alpha=\sqrt{\frac{\pi}{2}}$ for which the characteristic function was given in \eqref{eq:4cat} would be suitable candidates as well. To my knowledge there are no systems that directly exhibit such Hamiltonians. Albeit driven-dissipative engineering was used in \citep{Cohen_2017, cohen:tel-01545186} to effectively obtain $H_{4, cat}$, the same approach cannot be used here, as the steady states of the driven-dissipation process are incompatible with the targeted GKP-states.}

 The substrate-Hamiltonian considered here is thus
\begin{equation}
H_{sub}/E_J=-e^{-\pi}\sum_n L_n(2\pi)\ket{n}\bra{n}, \label{eq:H_sub}
\end{equation}
with Laguerre polynomials $L_n(\cdot)$, to which I apply the dynamical decoupling sequence introduced above. As described in eq. \eqref{eq:twirl_hamiltonian}, the dynamical decoupling sequence approximately projects the characteristic function of $H_{sub}$ (see fig. \ref{fig:hchar}) onto stabilizer-displacements. Since the substrate-Hamiltonian is given in the rotating frame, the displacements applied in the lab frame need to be adapted to account for that rotation. 

The first two Eigenstates of the effective average Hamiltonian under $N-level$ logical Twirl $H^{(0)}_{av}$  for $N=1, 5, 10, 15$ and $E_J=1$ are shown in fig. \ref{fig:hJJ_approx}. They are found to approximate the GKP magic states

\begin{align}
\ket{H^+_{\Delta}}&=\cos{\left(\frac{\pi}{8}\right)}\ket{0_{\Delta}}+\sin{\left(\frac{\pi}{8}\right)}\ket{1_{\Delta}},\\
\ket{H^-_{\Delta}}&=-\sin{\left(\frac{\pi}{8}\right)}\ket{0_{\Delta}}+\cos{\left(\frac{\pi}{8}\right)}\ket{1_{\Delta}},
\end{align}
with approximation parameter $\Delta_{q/p} \propto N^{-0.185}$.

\begin{figure}
\includegraphics{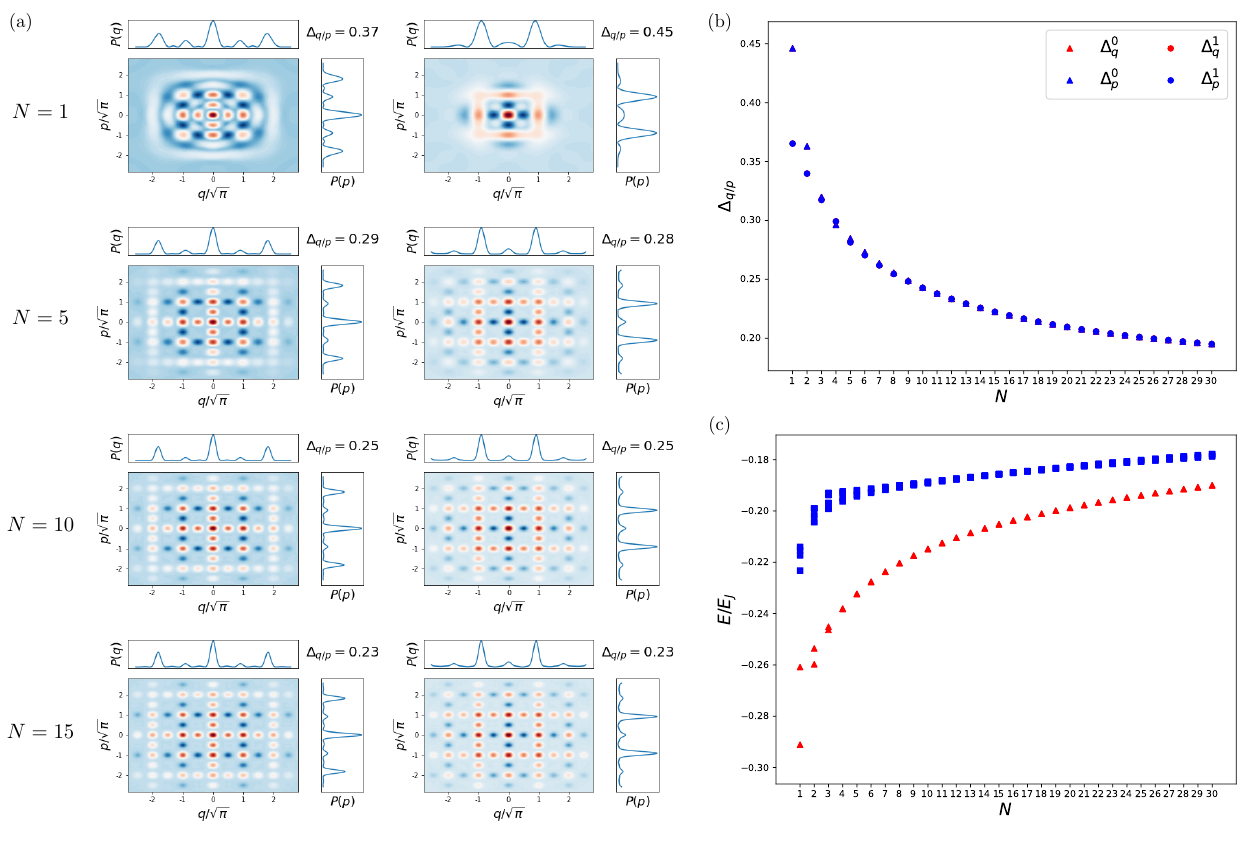}
\caption{$(a)$ Wigner functions of the two lowest eigenstates of $H^{(0)}_{av}$ for twirling level $N=1,..,15$ together with their effective squeezing parameter are shown. $(b)$  Finite squeezing parameters for $N=1..30$  and $(c)$  the ten lowest eigenenergies of $H^{(0)}_{av}$ are plotted. The two lowest (degenerate for large $N\geq 4$) eigenenergies approximating the GKP $\ket{\overline{H_{\pm}}}$ states are coloured red and are separated by a gap that shrinks with $N$ from the higher levels.
The effective squeezing of the two lowest eigenstates becomes symmetric in $q,\,p$ for $N\geq 8$ and approximately scales with the twirling-level $\Delta_{q/p}\propto N^{-0.185}$.\\
Code that was used to produce this figure can be found under \url{https://github.com/JonCYeh/GKP_DD}.
} \label{fig:hJJ_approx}
\end{figure}

Interestingly, in \ref{fig:hJJ_approx} $(c)$ it can be observed that the effective Hamiltonian is \textit{gapped}, where the gap shrinks with increasing twirling level $N$. Since this is not the case for the perfect GKP Hamiltonian, the gap is an effect of the imperfection of the approximation which has similarly been observed in \cite{rymarz2020hardwareencoding}. In practical scenarios, one may choose $N$ for a suitable trade-off between low squeezing parameter and larger gap of the Hamiltonian.

This pathway to realize the GKP Hamiltonian can also be understood in the broader framework of \textit{Floquet engineering} \cite{Rubio_Abadal_2020, Mori_2018}, which can be connected to by comparing the \textit{Floquet-Magnus} expansion to the expression for the average Hamiltonian given earlier. It has been observed that prior to an ultimate heat-death of periodically driven quantum systems a regime called \textit{Floquet-prethermalization} emerges in which the state of the system can be well described by the Gibbs state of the static effective (average) Hamiltonian of the system evolution \cite{Kuwahara_2016, Mori_2018_rev, Shirai_2015, Mori_2016}. This regime is predicted to remain valid to  a time exponentially long in the drive-frequency $\tau_h=e^{O(T_C^{-1})}$, giving ground for passively protected encoding of a GKP qubit in the quasi-stationary  non-equilibrium state of a driven harmonic oscillator as in this proposal. 
I remark that, similar to the findings here, the authors of \cite{Liang_2018} found that a comparable translation symmetric phase-space lattice can also be obtained for the motion of a driven ultracold atom trapped in an harmonic potential, where the drive a priori has a phase-space periodic structure. I refer to the numerical investigation of the thermalization behaviour in \cite{Liang_2018}.

\subsection{Parameter Inequalities}
Here I briefly comment on the  the experimental parameters necessary to realize the protocol.

%

The RWA for $h_{JJ}^{(1)}$ is generally valid for sufficiently large cavity frequencies $\hbar \omega \gg E_J$ which could be considered smeared out in the limit $\omega t_{min} \geq 2\pi$, where $t_{min}=2^{-4N}T_C$ corresponds to the smallest time scale where this approximation needs to be valid. The latter inequality sets the lower bound for the period length of the stroboscopic evolution  as 
$T_C \geq \frac{2\pi}{{\omega}} 2^{4 N}$. For the average Hamiltonian to remain valid, it would be desirable to have a cavity with large $\omega$ to minimize this bound for some finite $N$.

The final limitation of this scheme is given by the speed limit for displacement operations. For the BB approximation to hold, it is necessary that displacements of amplitude $|\xi|=\sqrt{\pi}$ can be realized in a time $t_{disp}\ll 2^{-4N}T_C$.
This necessitates that $\sqrt{2} T_X \ll 2^{-4N}T_C$ where $T_X$ is the minimal time it takes to realize a (elementary) displacement by $|\xi|=\sqrt{\frac{\pi}{2}}$. Altogether, this imposes a bound of $T_X \ll\frac{1}{\sqrt{2}} \frac{2\pi}{\omega}$ on the speed to implement displacements.

As a generic example, I find that with a cavity frequency of $\frac{\omega}{2\pi}=5.26 GHz$  \cite{Campagne_Ibarcq_2020}, elementary displacements must be realizable in a time $T_X \ll 0.13 ns$. To not induce any additional capacitances to the circuit for implementation of the displacements, I expect that they are better to be implemented by inductively coupling a drive to the oscillator.
\\ 

\subsection{Quantum Computation and State Initialization}\label{sec:computation}
Here I outline a few ways to use the above protocol for GKP-protected quantum computation. It is already known that the $\ket{H_{\pm}}$ states together with Clifford operations and homodyne measurement form a universal set of resources for GKP quantum computation \cite{Yamasaki_2020}, such that the proposed system immediately yields a platform for protected universal GKP quantum computation by executing standard GKP gates on the oscillator in between or during the stroboscopic evolution steps. For displacement gates, e.g. $\overline{X}$, the DD BPD sequence may potentially also be adapted to yield a net non-trivial displacement for chosen evolution rounds, e.g. by choosing $\{P_k\}$ such that $P_0P_1... P_M =\overline{X}$. It will be important that gates can be executed sufficiently quickly in a time $t_g \ll T_C$ to maintain a clear separation of time-scales.
 
Similar to the generation of a stabilizer Hamiltonian, implementing a cavity with impedance such that $\varphi=\sqrt{\frac{\pi}{2}}$, and choosing the control pulses such that $Q^X_k =Q_{k(2n,m)}=D\left((2n+i m)\sqrt{\frac{\pi}{2}}\right)$ allows to distill effective average Hamiltonians $H^{(1)}_{av}\propto \overline{X}$ or similarly for $Q^Z_k =Q_{k(n,2m)}=D\left((n+i 2m)\sqrt{\frac{\pi}{2}}\right)$ Hamiltonian $H^{(1)}_{av}\propto \overline{Z}$ can be obtained. By the same mechanism as to approximate the GKP Hamiltonian this allows to approximate logical GKP Pauli-Hamiltonians. 

One possibility for state initialization would be to measure a selected GKP logical operator, e.g. $\overline{Z}$, when the system is thermalized to the groundspace by coupling the oscillator to an extra qubit-mode via controlled-displacement and reading out the qubit \cite{terhal_review}. As an alternative, I expect that one may also  interpolate between a stabilizer- and logical Pauli Hamiltonian by incorporating a $\cos( 2\varphi)$ element \cite{Smith_2020} in the circuit which implies a Hamiltonian characteristic function comprising a double ring with radii $\varphi_1=\sqrt{\frac{\pi}{2}}$ and $\varphi_2=\sqrt{2\pi}$.  On such a substrate Hamiltonian it will be sufficient to adapt the displacement shift-amplitudes to interpolate between stabilizer- and logical displacement support. By thermalizing the system starting with a logical Pauli Hamiltonian and then switching to the symmetric logical control displacements $\{Q_{k(n,m)}\}$ would then in principle allow to initialize the system in a logical Pauli eigenstate.

\section{Discussion \& Outlook}
In this work I have derived energy-constrained twirling operations for GKP quantum error correction that are experimentally implementable and could be used for the suppression of coherent error accumulation or GKP state-distillation.  Furthermore, I showed how the designed twirling protocol can be translated into a dynamical decoupling sequence that is used to engineer an effective GKP stabilizer Hamiltonian.

Unlike the traditional application of DD, coherent control here is used to effectively project an engineered substrate-Hamiltonian onto the GKP stabilizer Hamiltonian and not employed for the removal noise processes. The degenerate ground state of the engineered Hamiltonian has been shown to host a GKP qubit. Due to a non-zero gap to higher excitations which is proportional to the Josephson energy, this ground space is expected to be useful for fault tolerant quantum computation. 

Practical realizations of the proposed scheme would require implementation of high impedance cavity modes, which are still challenging  with typical values ranging at $Z <R_Q$ \cite{Cohen_2017}. However, in contrast to the requirements for an autonomously protected cat state as in \cite{Cohen_2017}, the necessary cavity mode impedance of $Z=2 R_Q$ in this proposal to distill a stabilizer Hamiltonian appears as rather modest. The caveat here lies in the necessity for implementing very fast displacements of the cavity mode (while keeping the impedance and frequency high), which needs to happen at least an order of magnitude faster than than what has been demonstrated in recent experiments \cite{Campagne_Ibarcq_2020}. The strength of the present proposal for the implementation of the GKP Hamiltonian is that only well known ``textbook" superconducting circuit elements are necessary.

It would be interesting to study whether the same effect on the Hamiltonian can be obtained via continuous drive without employing the  bang-bang control limit. Depending on the maximal frequency by which displacements can be implemented in a concrete device, one may also investigate the use of \textit{concatenated} dynamical decoupling techniques to further optimise the proposed protocol. In an appropriate architecture, concatenation may also be used to interleave a DD protocol for distillation of a stabilizer Hamiltonian with one to distill a logical Hamiltonian for state reset or the implementation of logical Pauli rotations.

Beyond the questions regarding practical implementations a more basic question concerns the thermalization behaviour of the proposed system which relates to its capabilities in passive quantum error correction.  It would be interesting to study in how far thermalization behaviour can be linked to a rigorous classification of passive error correction as was performed in \cite{lieu2020symmetry} for cat qubits, in systems with an engineered Hamiltonian such as the one proposed here. Finally, I hope that thinking about Floquet engineering in terms of projecting away unwanted terms from an -- easier to implement -- substrate Hamiltonian via a dynamical decoupling sequence could be a useful perspective to exploit more generally.



\begin{acknowledgements}
I would like to thank R. Alexander and B. Baragiola for stimulating discussions that inspired this project, P. Faist for helpful discussions and A. Ciani for sharing his expertise on superconducting circuits. Furthermore I thank  J. Eisert, F. Arzani, R. Alexander and B. Terhal for valuable feedback. I am grateful for the the development and maintenance of QuTiP \cite{qutip_1, qutip_2} which has been used for the numerical calculations in this paper.
\end{acknowledgements}

\bibliography{references}

\begin{widetext}
\section{Appendix}

\subsection{Derivation of the substrate Hamiltonian}
In this section I will outline the derivation of the Hamiltonian corresponding to the circuit in fig. \ref{fig:circuit} and how to arrive at \ref{eq:hJJ}. Using standard circuit quantization \cite{Girvin_LectureNotes,Ciani:765439} the Hamiltonian can be expressed in terms of flux $\Phi$ and charge $Q$ that satisfy $[\Phi,Q]=i\hbar$ as

\begin{equation}
H=\frac{Q^2}{2C}+\frac{\Phi^2}{2L}-E_J\cos\left(\frac{2\pi}{\Phi_0} \Phi\right),
\end{equation}
where $\Phi_0=\frac{h}{2e}=2e R_Q$ is the flux quantum.
Expressing the Hamiltonian in terms of the cavity frequency $\omega=\sqrt{LC}^{-1}$ and creation and annihilation operators
\begin{align}
a &=\frac{1}{\sqrt{2L\hbar\omega}}\Phi+\frac{i}{\sqrt{2C\hbar\omega}}Q,\\
a^{\dagger} &=\frac{1}{\sqrt{2L\hbar\omega}}\Phi-\frac{i}{\sqrt{2C\hbar\omega}}Q,
\end{align}
which obey the commutation relation $[a,a^{\dagger}]=1$
the flux and charge operators can be expressed as
\begin{align}
\Phi&=\sqrt{\frac{\hbar Z}{2}} (a+a^{\dagger}) \\
Q &=-i\sqrt{\frac{\hbar}{2Z}} (a-a^{\dagger}),
\end{align}
where $Z=\sqrt{\frac{L}{C}}$ is the impedance of the cavity mode.
In this representation the Hamiltonian becomes
\begin{equation}
H= \underbrace{\hbar \omega\left(a^{\dagger}a+\frac{1}{2}\right) }_{H_0}-E_J \cos\left( \underbrace{\frac{2\pi}{\Phi_0} \sqrt{\frac{\hbar Z}{2}}}_{\varphi} (a+a^{\dagger}) \right).
\end{equation}
The constant factor inside the $cos(\cdot)$ term can be simplified to
\begin{equation}
\varphi=\sqrt{\frac{(2\pi)^2}{\frac{h^2}{(2e)^2}} \frac{\hbar Z}{2}}=\sqrt{\frac{\pi Z}{R_Q}}.
\end{equation}
In the frame rotating with $H_0$ the Hamiltonian reads
\begin{align}
H &=-E_J \cos\left( \varphi( e^{-i\omega t}a+e^{i\omega t}a^{\dagger}) \right)\\
&=-\frac{E_J}{2} \left\{  \exp( i\varphi e^{i\omega t}a^{\dagger}+i\varphi e^{-i\omega t}a  )+\exp( -i\varphi e^{i\omega t}a^{\dagger}-i\varphi e^{-i\omega t}a  ) \right\} \\
&=-\frac{E_J}{2} \left\{ D(i\varphi e^{i\omega t}) + D^{\dagger }(i\varphi e^{i\omega t})  \right\}
\end{align}

such that the Hamiltonian characteristic function in the rotating frame becomes
\begin{equation}
h_{rot}(\alpha; t)=-\frac{E_J}{2} \left\{ \delta^2(\alpha-i\varphi e^{i\omega t})+\delta^2(\alpha+i\varphi e^{i\omega t})\right\}.
\end{equation}

\subsection{Displacement Representation of the Photon Loss Channel}\label{appendix:photonloss}

Here I sketch the main steps of the derivation of the displacement representation of the photon loss channel \eqref{eq:c_photonloss}.
Key simplification to the following calculation is given by following fact to rearrange complex Gaussian integration, using the unitary matrices $$U=\frac{1}{\sqrt{2}}\begin{pmatrix} 
1 & i \\
1 & -i 
\end{pmatrix},\hspace{1cm}U_n=\oplus_{i=1}^n U. $$ 

\begin{fact}\label{lem}
Let $\vec{\alpha}=(\alpha_1,\alpha_1^*,\alpha_2,\alpha_2^*,...)^T,\vec{\beta}=(\beta_1,\beta_1^*,\beta_2,\beta_2^*,...)^T \in \mathbb{C}^{2 n}$ be vectors of complex numbers and their complex conjugate,  and $A=A^T \in \mathbb{C}^{2n\times 2n}:\; \mathrm{Re}(U_n^T A U_n)>0$ a symmetric complex matrix, for which the real part of $U_n^T A U_n$ is positive definite.\\
Then it holds that
\begin{equation}
\int d\vec{\alpha}\; e^{-\frac{1}{2}\vec{\alpha}^T A \vec{\alpha}+ \vec{\beta}^T\vec{\alpha} }=\sqrt{\frac{(2\pi)^n}{\det(A)}} e^{\frac{1}{2}\vec{\beta}^T A^{-1} \vec{\beta}}.
\end{equation}
\end{fact}

\begin{proof}
Since $\mathrm{Re}(U_n^T A U_n)>0$, $det(A)\neq 0$, $A$ is invertible. By quadratic completion the expression simplifies to 
 $$\int d\vec{\alpha}\; e^{-\frac{1}{2}\vec{\alpha}^T A \vec{\alpha}+ \vec{\beta}^T\vec{\alpha} }=e^{\frac{1}{2}\vec{\beta}^T A^{-1} \vec{\beta}}\; \int d\vec{\alpha}\; e^{-\frac{1}{2}(\vec{\alpha}-A^{-1}\beta)^T A (\vec{\alpha}-A^{-1}\beta )}=e^{\frac{1}{2}\vec{\beta}^T A^{-1} \vec{\beta}}\; \int d\vec{\alpha}\; e^{-\frac{1}{2}\vec{\alpha}^T A \vec{\alpha}}.$$
Since $U_n$ is unitary, the integration can be substituted by $\vec{\alpha}=U_n \vec{x},\; \vec{x}\in \mathbb{R}^{2n}$ with Jacobi determinant $|det(U_n)|=1$,
$$ \int d\vec{\alpha}\; e^{-\frac{1}{2}\vec{\alpha}^T A \vec{\alpha}}= \int d\vec{x}\; e^{-\frac{1}{2}\vec{x}^T B \vec{x}},$$ where $B=U_n^T A U_n=B^T$ is a complex symmetric matrix with positive definite real part. It is known, that for all real positive definite symmetric matrices $C$, $$\int d\vec{x}\; e^{-\frac{1}{2}\vec{x}^T B \vec{x}} = \sqrt{\frac{(2\pi)^n}{\det(C)}}.$$ Since the cone of all positive definite real matrices is convex, this equality holds on an open and connected subset of $ \mathbb{C}^{2n\times 2n}\sim \mathbb{C}^{4n}$ and therefore extends by the identity principle of complex analysis \cite{Lieb2001}.
\end{proof}

The continuous process matrix of the photon loss channel is found by first evaluating the displacement-coefficients of the Kraus operators.
\begin{align*}
c_l^{\gamma}(\delta)
&=Tr[\hat{D}^{\dagger}(\delta)\hat{E}_l] \\
&=\left(\frac{\gamma}{1-\gamma}\right)^{\frac{l}{2}} \frac{1}{\sqrt{l!}}\frac{1}{\pi^2}\int d^2\alpha d^2\beta \; \underbrace{\bra{\alpha}\hat{D}^{\dagger}(\delta)\hat{a}^l\ket{\beta}}_{\beta^l \braket{\alpha+\delta|\beta}e^{\omega(\alpha,\delta)/2}} \underbrace{\bra{\beta} (1-\gamma)^{\frac{\hat{n}}{2}} \ket{\alpha}}_{e^{-\frac{\gamma}{2}|\alpha|^2}\braket{\beta|\sqrt{1-\gamma}\alpha}} \\
&=\left(\frac{\gamma}{1-\gamma}\right)^{\frac{l}{2}} \frac{1}{\sqrt{l!}}\frac{1}{\pi^2} \int d^2 \alpha d^2\beta \; \beta^l e^{-|\alpha|^2-|\beta|^2+\alpha^*\beta+\sqrt{1-\gamma}\beta^* \alpha +\delta^* \beta - \alpha^* \delta -\frac{1}{2}|\delta|^2}.
\end{align*}
The exponent can be rearranged by introducing $\vec{\alpha}=(\alpha,\alpha^*,\beta,\beta^*)^T$, such that the integral is rewritten to
\begin{align}
\hspace{-1.3cm}...&=\left. \frac{1}{4}\frac{\partial^l}{\partial (\epsilon^*)^l}\right|_{\epsilon^*=0} \int d\vec{\alpha} e^{-\frac{1}{2}\vec{\alpha}^T A \vec{\alpha}
+\vec{J}^T\vec{\alpha}} e^{-\frac{1}{2}|\delta|^2}, \label{goto}
\end{align}
with
\begin{equation*}
A=
\begin{pmatrix} 
0 & 1 & 0 & -\sqrt{1-\gamma} \\
1 & 0 & -1 & 0 \\
0 & -1 & 0 & 1 \\
-\sqrt{1-\gamma} & 0 & 1 & 0 
\end{pmatrix}, \hspace{1cm} \vec{J}=(0,-\delta,\delta^*+\epsilon^*,0)^T.
\end{equation*}

It holds that $det(A)=(1-\sqrt{1-\gamma})^2=:\overline{\gamma}^{-2}$, with rescaled loss time $$\overline{\gamma}=\frac{1}{1-\sqrt{1-\gamma}}=\frac{2}{\gamma}-\frac{1}{2}-\frac{\gamma}{8} +O(\gamma^2)\in [1,\infty). $$

Using this rescaled loss time, the inverse of $A$ becomes
\begin{equation*}
A^{-1}=
\begin{pmatrix} 
0 & \overline{\gamma} & 0 & \overline{\gamma} \\
\overline{\gamma} & 0 & \overline{\gamma}-1 & 0 \\
0 & \overline{\gamma}-1 & 0 & \overline{\gamma} \\
\overline{\gamma} & 0 & \overline{\gamma} & 0 
\end{pmatrix},\hspace{1cm} \overline{\gamma}-1=\frac{\sqrt{1-\gamma}}{1-\sqrt{1-\gamma}}.
\end{equation*}
By fact \ref{lem}, the integration in eq. \eqref{goto} therefore evaluates to

\begin{align*}
\hspace{-1.3cm}..&=\left. \pi^2 \overline{\gamma} \frac{\partial^l}{\partial (\epsilon^*)^l}\right|_{\epsilon^*=0} e^{-(\overline{\gamma}-1)|\delta|^2-(\overline{\gamma}-1)\delta\epsilon^*-\frac{1}{2}|\delta|^2} \\
&=\pi^2 \overline{\gamma}\,(-(\overline{\gamma}-1)\delta)^l\, e^{-\frac{1}{2}(2\overline{\gamma}-1)|\delta|^2}.
\end{align*}

In total, this yields 
\begin{equation}
c_l^{\gamma}(\delta)=\frac{1}{\sqrt{l!}} \left(\frac{\gamma}{1-\gamma}\right)^{\frac{l}{2}}\overline{\gamma}\, (-(\overline{\gamma}-1)\delta)^l\,e^{-\frac{1}{2}(2\overline{\gamma}-1)|\delta|^2}. 
\end{equation}
The continuous process matrix is therefore evaluated to

\begin{align}
c^{\gamma}(\delta,\xi)
&=\frac{1}{\pi^2} \sum_l c_l^{\gamma}(\delta)c_l^{\gamma *}(\xi) \nonumber\\
&=\left(\frac{\overline{\gamma}}{\pi}\right)^2 e^{-\frac{1}{2}(2\overline{\gamma}-1)(|\delta|^2+|\xi|^2)}\sum_l \frac{1}{l!}\left[\left(\frac{\gamma}{1-\gamma}\right)(\overline{\gamma}-1)^2\delta\xi^* \right]^l \nonumber\\
&=\left(\frac{\overline{\gamma}}{\pi}\right)^2 e^{-\frac{1}{2}(2\overline{\gamma}-1)\left[|\delta|^2+|\xi|^2-2\xi^*\delta \right]}.
\end{align}
Using the identity for coherent states $$\braket{\xi|\delta }=e^{-\frac{1}{2}|\delta-\xi|^2-\omega(\xi,\delta)/2}=e^{-\frac{1}{2}(|\xi|^2+|\delta|^2-2\xi^* \delta)}$$ the last term can be expressed as power of the inner product of coherent states at $\xi,\,\delta$.

\subsection{Recompiling $\tau^N$}\label{appendix:recompile}
For $N$ repetitions of the local twirl, the displacements of consecutive rounds can be combined and  indexed by $(n,m) \in \{-N,...,N\}^2$. The resulting probability distribution of displacements can be interpreted as one of a $N$-step random walk starting at $(0,0)$ on a square lattice, where the vertices represent the labels $(n,m)$. By construction, this $2D$ random walk can be viewed as the Cartesian product of two $1D$ random walks $X$, see fig. \ref{fig:rdwalkX}, such that the probability to end up on vertex $(n,m)$ after $N$ steps decomposes as 
\begin{equation}
P_X^N(n,m)=P_X^N(n)P_X^N(m).    
\end{equation}
Each step in random walk $X=\sum_i X_i$ follows the probability distribution 
\begin{equation}
X_i=
\begin{cases}
1 &\, p=\frac{1}{4} \\
0 & \, p=\frac{1}{2} \\
-1 &\, p=\frac{1}{4}.
\end{cases}
\label{eq:RWX}
\end{equation}
A single step of random walk $X$ is equivalent to $2-$steps of a random walk $Y=\sum_j Y_j$ with half the step-length, i.e.
\begin{equation}
Y_i=
\begin{cases}
\frac{1}{2} &\, p=\frac{1}{2} \\
-\frac{1}{2} &\, p=\frac{1}{2}.
\end{cases}
\label{eq:RWY}
\end{equation}
Denoting the number of positive (negative) steps by $Y^+ = |\{i: Y_i>0 \}|$ ($Y^- = |\{i: Y_i<0 \}|$), the probability to end up on vertex $k=\frac{Y^+-Y^-}{2}$ after $N'=Y^++Y^-$ is given by 
\begin{equation}
P(Y_{N'} =k)=\frac{1}{2^{-N'}} \binom{N'}{Y^+}=\frac{1}{2^{-N'}} \binom{N'}{\frac{2k+N'}{2}}.\label{eq:dispN_prob}
\end{equation} 
One can check that $P(Y_2=k)$ reproduces the probabilities given in \eqref{eq:RWX}.
Finally, this yields
\begin{equation}
P_X^N(n,m)=2^{-4 N}\binom{2N}{n+N}\binom{2N}{m+N},
\end{equation}
which gives a closed expression for the  twirl-measure
\begin{equation}
d\mu_N(\gamma)=\sum_{n, m =-N }^N P_X^N(n,m) \delta^{2}\left(n\sqrt{\frac{\pi}{2}}+im\sqrt{\frac{\pi}{2}}-\gamma\right) d^2\gamma. \label{eq:measure_N}
\end{equation}

\end{widetext}

\end{document}